\documentclass[a4paper,12pt, notitlepage,fleqn]{article}

\usepackage{graphicx}
\usepackage[usenames, dvipsnames]{xcolor}
\usepackage{verbatim}
\usepackage[latin1]{inputenc}
\usepackage{fancyhdr}
\usepackage[T1]{fontenc}
\usepackage{epsfig}
\usepackage{subfigure}
\usepackage{amsmath}
\usepackage{amssymb}
\usepackage[amsmath,thmmarks]{ntheorem}
\usepackage{bbm}
\usepackage{mathrsfs}
\usepackage{enumitem}
\usepackage[longnamesfirst, sort]{natbib}
\usepackage{pgfplots}
\usetikzlibrary{pgfplots.groupplots}
\usetikzlibrary{external} 
\pgfplotsset{compat=newest}
\usetikzlibrary{patterns}
\usetikzlibrary{intersections}
\usepgfplotslibrary{fillbetween}
\usepackage{accents}
\usepackage{setspace}
\usepackage[hidelinks]{hyperref}

\usepackage{sgame}

\newcommand{\bmat}[1]{\left[\begin{array}{#1}}
	\newcommand{\emat}{\end{array}\right]}

\allowdisplaybreaks[1]

\newtheorem{corollary}{Corollary}

\newtheorem{lemma}{Lemma}

\newtheorem{proposition}{Proposition}

\newenvironment{proof}[1][Proof]{\noindent\textsc{#1:} }{\ \rule{0.5em}{0.5em}}

\begin{document}\onehalfspacing
	
		\title{Never-ending Search for Innovation}
	
	\author{Jean-Michel Benkert and Igor Letina\thanks{Benkert: Department of Economics, University of Bern. Letina: Department of Economics, University of Bern and CEPR. We thank Georg N\"{o}ldeke for very valuable comments and Martial Zurbriggen for excellent research assistance.}}

    \date{February 2025}
	\maketitle
	
	\begin{abstract} 

\noindent We provide a model of investment in innovation that is dynamic,  features multiple heterogeneous research projects of which only one potentially leads to success, and in each period, the researcher chooses the set of projects to invest in. We show that if a search for innovation starts, it optimally does not end until the innovation is found---which will be never with a strictly positive probability.

\vspace*{0.6cm}
		
		\noindent \textbf{Keywords:} Innovation, Optimal Search, Infinite Horizon

\vspace*{.1cm}

		\noindent \textbf{JEL Codes:}  D$83$, O$31$ 
	\end{abstract}

\vspace*{.1cm}
 
	\newpage

\section{Introduction}

In 1880, Alphonse Laveran discovered that malaria was caused by a parasite from the \emph{Plasmodium} genus. This discovery immediately led to the next question: how was \emph{Plasmodium} transmitted from one person to another? To answer such a question, researchers could design and conduct various experiments to test the many potential vectors, including transmission via soil, water, or a living host such as a bird or---which ultimately proved true---a mosquito \citep{desowitz1991malaria}. For example, researchers could expose healthy subjects to infected water to determine if the disease could be transmitted that way, or allow mosquitoes to feed on infected individuals and then examine the insects for the presence of the parasite.\footnote{These are some of the experiments conducted by Ronald Ross, who in 1897 demonstrated that mosquitoes transmitted malaria.} There are countless experimental variations, such as the species of mosquito used, the time elapsed after infection before examination, the specific parts of the mosquito dissected, and the techniques employed to detect the parasite. Researchers, aiming to find the answer as efficiently as possible, would pursue the most promising and least costly experiments first, and move on to more challenging, expensive, or less likely experimental avenues if success initially proves elusive. 

The above example has two important features. First, given a set of potential vectors, discovering that one is not responsible for transmission is informative about the probability that a remaining vector is how the disease actually spreads. For instance, if the parasite is not found in the soil, the researcher may update her belief that the available methods are unable to detect the parasite, leading her to become more pessimistic about the probability of finding the parasite in, for example, water. At the same time, conditional on being able to detect the parasite and the parasite being present in one environment only, then failing to detect the parasite in soil increases the belief of the researcher that she will detect it in water. Second, the cost of testing different vectors varies greatly. For example, examining soil samples is cheap, while dissecting living hosts is costly. 

We propose a model exhibiting both of the above features, informational externalities and heterogeneous projects. In doing so, we go beyond the existing literature and we find that, if the search starts, it does not end until the answer is found. In particular, the search will go on forever with strictly positive probability. This is in sharp contrast with the existing results in the literature, where the search tends to end in finite time even if the optimal search rule prescribes to continue searching forever.\footnote{This holds true in the seminal search models (without informational externalities) such as \citet{mccall1965economics} and \citet{morgan1983search}. See \citet{lippman1976economics} and \citet{chade2017sorting} for surveys, or for more recent examples see \citet{benkert2018optimal} and \citet{mccardle2018abandon}. In this class of models, even if the optimal search rule prescribes searching forever, the probability that search continues forever will converge to zero, thus ending search in finite time.} 

We study a model where the agent is faced with a set of research projects of which only one has the potential to lead to the breakthrough. Importantly, our model features learning, as failure from one set of projects is informative about the chances that a success will be found in the remaining set. Our model is closely related to \citet{deroos2018shipwrecks}, who also consider a model where an agent searches in a set of research projects. Importantly, they assume that all projects are homogeneous, and that a successful project exists with certainty. When this is the case, then the optimal search prescribes that the entire set of available projects is examined in finite time, causing the search to also end in finite time.\footnote{Using a similar model as \citet{deroos2018shipwrecks}, \citet{matros2014treasure} and \citet{matros2016duplicative} focus on the strategic interaction between multiple searchers. In extensions, they consider heterogeneous projects (in terms of success probabilities rather than costs) and allow for no successful project to exist. Nevertheless, they always conclude that search ends in finite time, contrasting with our findings.} We instead assume that projects are sufficiently heterogeneous, so that a rational agent optimally never examines the entire set of available projects. Paired with the possibility that no project could be successful, we obtain the contrasting result that optimal search may never end. 

Heterogeneity of research project is not sufficient for the search to potentially remain active forever. This can be best seen in \citet{weitzman1979optimal}, who famously showed that optimal search ends---in finite time---when some reservation value is met. In \citet{weitzman1979optimal}, opening one box is not informative about the value of yet unopened boxes. Thus, after the agent has exhausted all the promising boxes, she stops. In contrast, Bayesian updating leads our agent to become more optimistic about the value of remaining projects, and, as we show, this causes the agent to never willingly abandon an active search.  

It is instructive to compare our framework with the class of bandit models, which are often used to model dynamic innovation.  In a standard bandit model,\footnote{As an example of such a model used in the context of innovation see \cite{bergemann2005financing}.} the searcher faces a project (or an arm) with an unknown return. The project is either good, so that investing in it yields a success with some probability, or bad, so that any investment is wasted.  In these models, the absence of a success is informative and leads to increasing pessimism: as failures accumulate, the searcher lowers her belief about the probability that she is facing a good project, until she becomes so pessimistic that she stops searching, which happens in finite time. Our model features (infinitely) many heterogeneous projects, and the structure of the model is such that becoming more pessimistic about one project necessarily leads the searcher to become more optimistic about the unexplored projects. This positive informational spillover across projects, together with project heterogeneity, generates a never-ending search in our model, while a standard bandit model results in search ending in finite time.

In what follows, we formally introduce our model (Section \ref{sec:model}), present our result on potentially never-ending search (Section \ref{sec:results}), and conclude with a discussion of the result and its policy implications (Section \ref{sec:conclusion}).

\section{Model}\label{sec:model}

There is a set of heterogeneous research projects $J= [0, 1)$ and a single agent searching for an innovation in $J$. In any period $t = 1, 2, \ldots$, the agent can examine an arbitrary (measurable) set of projects $S \subseteq J$ at the cost $ C(S) = \int_S c(j) dj$. We assume that $c(j)$ is continuous, strictly increasing, and satisfies $\lim_{j \rightarrow 1} c(j) = \infty$. 

With probability $p \in (0,1)$, the innovation is \emph{feasible}, meaning that there is a single project $\hat{j} \in J$, which, if examined, will lead to a successful innovation of value $v$. Examining any other project leads to a dead end. We assume that any project $j \in J$ is equally likely to yield success. The future costs and payoffs are discounted by $\delta\in (0,1)$.\footnote{This is a dynamic version of the model introduced in \citet{letina2016road}.}

\section{Never-ending Search} \label{sec:results}

An agent's search strategy $\sigma = (L_1, L_2, \ldots)$ consists of a sequence of (potentially empty) measurable sets $L_t\subseteq[0, 1)\cup \varnothing$, $t=1, 2, \ldots$ determining which projects will be examined in period $t$ if previous periods did not yield success. We say that a search strategy $\sigma$ is \emph{interval-based} if the subsets $L_t$ are intervals. Moreover, we say that a search strategy is \emph{increasing-interval-based} if the subsets $L_t$ are intervals such that $\sup L_t = \inf L_{t+1}$.

Let $S_t := \cup_{t' < t} L_{t'}$ and $U_t := [0,1) \setminus S_t$. Thus, $S_t$ represents the set searched up to time $t-1$, and $U_t$ is the set still unsearched when deciding at time $t$. Therefore, any optimal search in period $t$ will occur in $U_t$.

Let $\hat{p}_\sigma(t) := \Pr[\hat j \in U_t\mid \hat j  \not \in S_t]$ denote the agent's Bayesian posterior belief that success is feasible at time $t$ when following the strategy $\sigma$, i.e., $\hat{p}_\sigma(t) = \frac{\mu(U_t)p}{\mu(U_t)p + (1-p)}$, where $\mu$ denotes the Lebesgue measure. Then, we can write the agent's value function at time $t$, when $S_t$ has been searched so far and when searching according to strategy $\sigma = (L_1, L_2, \ldots)$, as
\begin{align*}
	V_\sigma(S_t, t) =  \hat{p}_\sigma(t) \frac{\mu(L_t)}{\mu(U_t)} v - C(L_t)+ \delta \left(1-\hat{p}_\sigma(t) \frac{\mu(L_t)}{\mu(U_t)}\right)V_\sigma(S_t\cup L_t,t+1).
\end{align*}

We can now state our result, the proof of which is relegated to the appendix.
\begin{proposition}\label{prop1}
	The optimal search strategy $\sigma^\ast$ is such that either $L_t= \emptyset$ for all $t$ or $\mu(L_t)>0$ for all $t$.
\end{proposition}

Intuitively, there are two forces leading to the result. First, the informational externality of not finding the innovation induces the agent to continue searching. Second, the increasing costs of searching dampen the incentives to continue the search. The former effect ensures that the search continues, whereas the latter effect (and in particular the assumption that $\lim_{j \rightarrow 1} c(j) = \infty$) prevents the search from being exhaustive. 

More specifically, if, contrary to our result, there were a final period $T$ in which the search occurred, the agent would want to search all projects with a positive expected value in that period, implying that the marginal project would have an expected value of zero. Because of the increasing costs of search, that marginal project would have to be interior. However, if there was no success in period $T$, the belief that the marginal project will succeed would increase, implying a strictly positive expected value in the neighborhood of the marginal project. Thus, stopping the search could not be optimal. Further, as the remaining projects are increasingly costly to search, exhaustively searching all remaining projects and ending the search in finite time is never optimal. Hence, once started, the search only stops if the innovation is discovered. Since the innovation is not feasible with probability $1-p>0$, the search could continue forever with a strictly positive probability, leading to the following corollary. 

\begin{corollary}
    If any search is optimal, the probability of active search in any period $t\geq1$ is strictly above $1-p>0$.
\end{corollary}
This corollary highlights the distinction of our result relative to the existing literature, where even if the optimal strategy is to search forever, the search ends in finite time. In contrast, our result is not only about the optimal strategy mandating infinite search but that the actual search may last forever. 

Finally, for our results to be meaningful, we need to establish conditions under which some search is indeed optimal.

\begin{lemma} \label{lem:search_atall}
    If $p v > c(0)$, it is optimal for the agent to engage in the search.
\end{lemma}
It is immediate that this condition ensures that the agent would rather search once than not search at all.\footnote{In the one-period model, searching an interval of length $q^\ast$ is optimal if and only if $pv=c(q^\ast)$. Hence, the condition in Lemma \ref{lem:search_atall} yields $q^\ast>0$, ensuring that the agent would rather search once than not search at all.} Intuitively, the agent compares the marginal benefit of searching the initial interval with the marginal cost of doing so. The condition ensures that the former exceeds the latter, so that search starts. However, by Proposition \ref{prop1} we know that the agent would not cease searching after one period but continue until a success is found.

\section{Conclusion}  \label{sec:conclusion}

At first glance, our result may seem counterintuitive: why would a rational agent be willing to keep (forever) searching for a prize of finite value at eventually unbounded marginal costs? To reconcile our result with this (incomplete) intuition, we need to also consider the search \emph{intensity}. Even as marginal costs increase, the agent can adapt the intensity of search so that actual search costs do not increase from period to period. Effectively, while active search may never end, the associated search intensity will converge to zero in the limit because of the unbounded increase in the marginal search costs. 

This is not merely a technical point, but rather has an important policy implication for the question of when to optimally abandon a (research) project. \citet{mccardle2018abandon} note that, in practice, search often lasts longer than what is \emph{perceived} as optimal and suggest that one should set an ``alarm clock'' to ensure a timely end to the search effort. However, our results indicate that such a policy may be short-sighted. What appears to be excessive search from an external perspective may in fact be optimal, as the agent internalizes the informational value of continued failures and adjusts search intensity accordingly. An outside decision-maker, who does not fully account for these learning dynamics, may wrongly conclude that search has gone on ``too long'' and intervene prematurely. This is particularly relevant in environments with informational spillovers (so that failure today is informative of success chances tomorrow) and adjustable research intensity, which we would argue are more common than not in research. Moreover, our result suggests that the debate over whether research lasts ``too long'' or stops ``too soon'' is too simplistic. Rather than a binary choice, search intensity can decline over time while still continuing indefinitely. To address this empirically, our finding highlights the need to track not just whether research is ongoing, but how its intensity evolves.

In short, the main message of our paper is that external decision makers should exercise caution when recommending that a research project be abandoned. What might appear to be the stubbornness of a researcher unwilling to admit defeat could, in fact, be a rational response to the lessons learned from past failures.


\newpage 

\section*{Appendix: Proof of Proposition \ref{prop1}}

\begin{lemma}\label{lem:nogaps}
The optimal search strategy $\sigma^\ast$ satisfies: there do not exist intervals $L$ and $L''$ with $\mu(L), \mu(L'') > 0$ and $\inf L > \sup L''$ such that $L$ is searched and $L''$ is not.
\end{lemma}

\begin{proof}
	Suppose not. Then, there exists a set $L$, searched at some time $t$ and a set $L'$, with $\inf L> \sup L'$ and $\mu(L') = \mu(L)$, that is never searched. As $L$ is searched at time $t$,  we can write $L_t = L\cup S$ where $S$ may be empty. Then, the value function at time $t$ when searching according to ${\sigma^\ast}$ is
	\begin{align*}
		V_{\sigma^\ast}(S_t, t) =  \hat{p}_{\sigma^\ast}(t)& \frac{\mu(L\cup S)}{\mu(U_t)} v - C(L\cup S) \\ &+ \delta \left(1-\hat{p}_{\sigma^\ast}(t) \frac{\mu(L\cup S)}{\mu(U_t)}\right)V_{\sigma^\ast}(S_t\cup L\cup S,t+1).
	\end{align*}
	Consider the alternative strategy $\sigma'$, coinciding with ${\sigma^\ast}$ except that $L$ is replaced by $L'$ at time $t$. Then
	\begin{align*}
		V_{\sigma'}(S_t,t) =  \hat{p}_{\sigma'}(t)& \frac{\mu(L'\cup S)}{\mu(U_t)} v - C(L'\cup S)  \\ &+ \delta \left(1-  \hat{p}_{\sigma'}(t)\frac{\mu(L'\cup S)}{\mu(U_t)}\right)V_{\sigma'}(S_t\cup L'\cup S,t+1).
	\end{align*}
	The difference in value functions $V_{\sigma'}(S_t,t)-V_{\sigma^\ast}(S_t, t)$ reads
    	\begin{align*}
& C(L\cup S) - C(L'\cup S)\\
        & +\delta \left(1-\hat{p}_\sigma(t)\frac{\mu(L'\cup S)}{\mu(U_t)}\right)\left(V_{\sigma'}(S_t\cup L'\cup S,t+1)- V_{{\sigma^\ast}}(S_t\cup L\cup S,t+1)\right)\\
		&= C(L\cup S) - C(L'\cup S)>0.
	\end{align*}
    The equality obtains because the deviation from ${\sigma^\ast}$ to $\sigma'$ entails no change in the posteriors, i.e., $\hat{p}_{\sigma'}(t') =\hat{p}_{{\sigma^\ast}}(t') $ for all $t'$, as the measures of the searched and unsearched sets are the same with both strategies. Hence, since the two strategies coincide for any $t'>t$, we have $V_{\sigma'}(S_t\cup (L'\cup S), t+1)= V_{{\sigma^\ast}}(S_t\cup (L\cup S), t+1)$. Finally, the inequality follows from the fact that $\inf L> \sup L'$ and $c$ is increasing. This contradicts the optimality of $\sigma^\ast$.
\end{proof}

\begin{lemma} \label{lem:increasing}
	The optimal search strategy $\sigma^\ast$ satisfies: there do not exist intervals $L$ and $L''$ with $\mu(L), \mu(L'') > 0$ and $\sup L < \inf L''$ such that $L$ is searched at time $t$ and $L''$ at time $t' < t$.
\end{lemma}
\begin{proof}
	Suppose not. Then, there exist sets $L$ and $L'$ that are searched at times $t$ and $t'$ with $t>t'$, $ \sup L < \inf L''$ and $\mu(L') = \mu(L)$. Consider the alternative strategy $\sigma'$ which coincides with $\sigma^\ast$ except that $L$ and $L'$ are swapped at $t$ and $t'$, respectively. Analogously to the proof of Lemma \ref{lem:nogaps}, the posterior beliefs and success probabilities induced by the two strategies are the same. Hence, the only difference in the value functions is induced by a change in the timing of the search costs. The swap reduces the search costs at $t'$ by $K$, while increasing them by the same amount at $t$, but this increase is discounted by at least $\delta$. Thus, the net reduction in costs is at least $K - \delta K$, implying $V_{\sigma'}(S_t, t) - V_{\sigma^\ast}(S_t, t) > 0$, contradicting the optimality of $\sigma^\ast$.
\end{proof}
Lemmas \ref{lem:nogaps} and \ref{lem:increasing} imply the following: 
\begin{corollary}
	The optimal search strategy $\sigma^\ast$ is increasing-interval-based.
\end{corollary}
With this we can reinterpret the strategy $\sigma=(L_1, L_2, \ldots)$ as a collection of increasing intervals, i.e., $[0, l_1), [l_1, l_2)$ and so forth.

\begin{lemma} \label{lem:no-breaks}
	The optimal search strategy $\sigma^\ast$ satisfies: if there exists a period $t$ with $\mu(L_t) = 0$, then $\mu(L_{t'}) = 0$ for all $t' > t$.
\end{lemma}
\begin{proof}
	Suppose not. Then, there exists a period $t$ with $\mu(L_t)=0$ and a period $t'>t$ with $\mu(L_{t'})>0$. Consider the alternative strategy $\sigma'$ which coincides with $\sigma^\ast$ except that for all $s\geq t$ we replace $L_{s}$ by $L_{s+1}$. Then, $
		V_{\sigma'}(S_{t}, t) =  V_{\sigma^\ast}(S_{t+1},t+1)
	$
	and
	$
		V_{\sigma^\ast}(S_{t},t) = \delta V_{\sigma^\ast}(S_{t+1},t+1).
	$
Since stopping search yields a payoff of zero, it cannot be that $ V_{\sigma^\ast}(S_{t''},t'') <0$ for any $t''$. Additionally, since $\mu(L_{t'})>0$, it must be that $V_{\sigma^\ast}(S_{t'},t') > 0$, since otherwise removing some subset from $L'$ would strictly increase the payoff in $t'$. Thus, $V_{\sigma^\ast}(S_{t+1},t+1) > 0$ and since $\delta < 1$, then $V_{\sigma'}(S_{t}, t) > V_{\sigma^\ast}(S_{t},t)$.
\end{proof}

\vspace{0.5cm}
\noindent We can now prove Proposition \ref{prop1}.

\begin{proof}
	Lemma \ref{lem:no-breaks} implies that only three types of strategies can be optimal: 1) never search, 2) search in every period until time $T>0$ (unless there's a success) and then not search anymore, 3) search in every period until there is a success. Thus, we only need to show that 2) cannot be optimal. Suppose there is a $T>0$ such that the agent searches in every period $t\leq T$ and then stops. Then, the agent's value function at time $T$, the last period of search, is given by
	\begin{align*}
		V_{\sigma^\ast}(S_T, T)
  &= \frac{(l_T-l_{T-1})p}{(1-l_{T-1})p + (1-p)}  v - C([l_{T-1}, l_{T} )),
	\end{align*}
	because $V_{\sigma^\ast}(S_{T}\cup [l_{T} - l_{T-1}),T+1)=0$ and by Bayes' rule
    $
		\hat{p}_{\sigma^\ast}(T) = \frac{(1-l_{T-1})p}{(1-l_{T-1})p + (1-p)}.
    $
	Further, because it is optimal to search in period $T$, $l_T$ is implicitly defined by\footnote{Observe that the value function is strictly concave in $l_T$ as $c$ is strictly increasing.}
	\begin{align}
		\frac{\partial{V_{\sigma^\ast}(S_t, T)}}{{\partial l_T}} = 0 
		&\Leftrightarrow \frac{pv}{(1-l_{T-1})p + (1-p)} = c(l_{T}).\label{eq:last_period_search}
	\end{align}
Since $\lim_{j \rightarrow 1} c(j)= \infty$ and the LHS of equation \eqref{eq:last_period_search} is finite, it cannot be optimal that $l_{T}=1$ and the interval $[0,1)$ is searched in finite time.
 
	Thus, take $l_{T}<1$ and consider the agent's incentives in period $T+1$. In particular, consider one more period of search before stopping forever. Then, the value function of searching a positive interval, i.e., setting $l_{T+1}>l_T$ is
	\begin{align*}
		V_{\sigma'}(S_{T+1}, T+1)
		&= \frac{(l_{T+1}-l_{T})pv}{(1-l_{T})p + (1-p)}- \int_{l_T}^{l_{T+1}}c(j)dj.
	\end{align*}
	For $\sigma^\ast$ to be optimal, we need to have
	$
		\frac{pv}{(1-l_{T})p + (1-p)}-c(l_{T+1})<0, 
	$
	i.e., the marginal value of continuing search needs to be negative. Rearranging and using equation \eqref{eq:last_period_search}, we obtain
	\begin{align*}
			&	\frac{pv}{c(l_{T+1})} < 1-l_{T}p
			\Leftrightarrow \frac{c(l_{T})}{c(l_{T+1})}  < \frac{1-l_{T}p}{1-l_{T-1}p},
	\end{align*}
	which must hold for any $l_{T+1}\in (l_T, 1)$. This is violated for $l_{T+1}$ close enough to $l_T$, as the LHS converges to $1$ while the RHS is constant and strictly smaller than $1$ as $l_{T+1}$ approaches $l_T$ from above. Hence, there is an incentive to continue searching in period $T+1$; thus, stopping in $T$ cannot be optimal.
\end{proof}

\newpage

\bibliographystyle{ecta}
\bibliography{optimal_search}

\end{document}